\documentclass[12 pt]{article}
\usepackage[utf8]{inputenc}
\usepackage[fleqn]{amsmath}
\usepackage{array}
\usepackage{appendix}
\usepackage{romannum}
\usepackage{amsfonts}
\usepackage{enumitem}
\usepackage{soul}
\usepackage{amssymb,latexsym}
\usepackage{graphics}
\usepackage{float}
\usepackage{subcaption}
\usepackage{graphicx}
\usepackage{comment}
\usepackage{mathtools}
\usepackage{amsfonts}
\usepackage{amsthm}
\newtheorem{theorem}{Theorem}

\usepackage{multicol}
\usepackage{cite}
\usepackage{tikz}
\usepackage{amsmath, amsthm, amscd, amsfonts, amssymb, graphicx}
\usepackage[left=1.5 cm,right=1.5 cm,top=1.5cm,bottom=1.5cm]{geometry}
\title{A discrete-time dynamical model of prey and stage-structured predator with juvenile hunting incorporating negative effects of prey refuge}
\author{Debasish Bhattacharjee$^1$, Nabajit Ray$^2$, Dipam Das$^3$,  Hemanta Kumar Sarmah$^4$ \\
$^{1,3,4}$Department of mathematics, Gauhati University, Assam, India  \\
$^{2}$Department of Mathematics, S.B. Deorah College, Assam, India\\
Emails:$^1$debabh2@gmail.com, $^2$nabajitray@gmail.com, \\ 
$^{3}$pontu.dd@gmail.com, 
$^{4}$hsarmah@gauhati.ac.in
}

\date{}

\begin{document}
\maketitle
\begin{abstract}
   This paper examines a discrete predator-prey model that incorporates  prey refuge and its detrimental impact on the growth of the prey population. Age structure is taken into account for predator species. Furthermore, juvenile hunting as well as prey counter-attack are also  considered. This paper provides a comprehensive analysis of the existence and stability conditions pertaining to all possible fixed points. The analytical and numerical investigation into the occurrence of different bifurcations, such as the Neimark-Sacker bifurcation and period-doubling bifurcation, in relation to various parameters is discussed. The impact of the parameters reflecting prey growth and prey refuge is thoroughly addressed. Numerous numerical simulations are presented in order to validate the theoretical findings.

\end{abstract}

\textbf{2020 AMS Classifications:} 92D25, 92D40, 92D50, 92B05, 39A05, 39A28, 39A30\\


\pagenumbering{arabic}

\section{Introduction} One of the most  pivotal mechanism in maintaining the ecological balance of the ecosystem is the dynamic interplay between prey and predator. In recent years, mathematical models have gained significant traction and utility in explicating population dynamics. The predator-prey model has garnered significant attention from researchers in the field of ecology, following the groundbreaking contributions of Lotka \cite{lot} and Volterra \cite{vol} to the field. Subsequent to that, numerous kinds of enhancements for the predator-prey model have been suggested \cite{lon,ni,jan,yas,mat,yang}.

Predominantly, within predator-prey systems, it is commonly assumed that predators within a given population possess uniform predation capacity and fecundity. However, in the real world, several authors believed that predators residing within a specific population can be classed by two fixed ages: juvenile predator and adult predator. Numerous scholarly articles have been dedicated to examining the dynamics of populations structured by stages. In recent times, several researchers have employed stage structure in prey species \cite{zha,liu}, while numerous scholars have utilised stage structure in predator species \cite{xia,kal,Mor} as well. This paper will focus solely on the stage structure of predator species. Morever, during this juvenile phase, predators develop the essential predatory skills required for their survival. Predation during the juvenile phase poses a significant challenge, given that early juvenile predators lack the requisite abilities and expertise in foraging and hunting. In numerous instances, when juvenile predators partake in attacks on their prey, the prey species respond by initiating counter-attacks, leading to the killing of young and inexperienced predators as a means of self-defense\cite{mag,aok,dor}.  Only a small number of scholarly articles have explored the intricacies of juvenile hunting through the application of mathematical models\cite{motive1,motiverel1,motiverel2}. So, juvenile hunting and in response counter-attacks by the prey are also addressed in this work.

Numerous studies \cite{sk,dey} have confirmed that predators employ diverse tactics in order to capture their prey . Likewise, prey species employ diverse strategies\cite{shi,sing,preyrefuge1} to mitigate the rate of predation. Prey refuge is one of them. The utilisation of refuges can afford a certain degree of protection to prey species. Despite its benefits, prey refuge can have negative consequences on prey's growth. The utilisation of refuges by prey carries significant costs, particularly in terms of potential reductions in feeding or mating success due to increased time spent in refuges\cite{preyrefdisad}. As a result, the presence of prey refuge may result in a decrease in growth of prey\cite{done,mar,fra}. In this paper, we consider that prey uses refuge, which has a detrimental effect on their growth. 

In cases where there are populations with overlapping generations, the birth processes take place in a continuous manner. As a result, the interaction between predator and prey is typically represented through the use of ordinary differential equations \cite{muk}. However, it is worth noting that in reality, there are other types of species, such as monocarpic plants and semelparous animals \cite{kot}, that exhibit discrete non-overlapping generations and their births occurexclusively during regular breeding seasons. Their interactions are characterised by difference equations or as discrete-time mappings. The dynamics of discrete-time predator-prey models might reveal greater complexity compared to their continuous-time models\cite{intro1}. When the population size is relatively small, it is appropriate to use discrete models to represent populations, even if some species have a long lifespan and overlapping generations. Additionally, population change is typically examined on a yearly (or monthly, or daily) basis. Hence, it is imperative to examine discrete population dynamical models\cite{intro2}. In recent years, there has been a significant increase in collaboration among ecologists studying discrete dynamical ecological models \cite{recent1,recent2}. 

Kaushik et al. \cite{motive1} considered a mathematical model which is as follows:
\begin{equation}\label{omodel}
\begin{cases}
    \frac{dx}{dt}=rx(1-\frac{x}{k})-\alpha_1 x y- \alpha_2 x z\\
    \frac{dy}{dt}=\mu \alpha_2 x z-\alpha_3 x y-\beta y-\gamma y\\
    \frac{dz}{dt}=\gamma y-m z-\phi z^2
    \end{cases}
\end{equation}
where $x$, $y$, $z$ represent population sizes of prey, juvenile predator, and adult predator respectively.

The objective of this study is to analyse the aforementioned model, taking into account the prey's refuge behaviour and its detrimental impact on prey population growth, within the context of a discrete-time framework. This analysis aims to provide an in-depth investigation of the advantages and disadvantages associated with the prey refuge, a topic that has not been thoroughly explored in the existing scholarly literature. To the best of the authors' knowledge, there has been no prior investigation into a mathematical model that elucidates the adverse effects of prey refuge on the population dynamics of prey.

This paper is structured in the subsequent fashion: section 2 presents the mathematical framework of the system. Section 3 presents the parametric conditions that pertain to the existence and stability of the equilibrium points. The theorems pertaining to Neimark-Sacker bifurcation and period-doubling bifurcation are presented in sections 4 and 5. In section 6, Numerical simulations are provided to validate the analytical results. Finally, section 7 draws a quick conclusion.

\section{Mathematical Modelling}
Kaushik et al. \cite{motive1} examined a mathematical model that describes the dynamics of a prey species and a predator species with stage structure, as presented in equation (\ref{omodel}). It is postulated that both juvenile and adult predators exhibit Holling type I functional response when interacting with prey species. Morever, it is considered that juvenile predators cannot reproduce, only adult predators can.
The model (\ref{omodel}) is further modified to incorporate the anti-predator effect, specifically the utilisation of prey refuge. We assume that prey uses prey refuge to mitigate adult predator attacks. The prey refuge effect has no impact on juvenile hunting.
 This may be due to the fact that the size structure or desire to hunt causes juveniles to exert extra effort to thwart the prey refuge effect, or because the prey species has no fear response to juvenile hunting, and thus the anti-predator behaviour towards juvenile predators is not the prey refuge but the aggressive counter-attack. 
 
  Let n represents the prey refuge constant, such that $nx$ represents the number of prey species that are unaccessible to the adult predator, and adult predators do not concern themselves with the pursuit of this quantity of prey. Consequently, $(1-n)x$ denotes the quantity of prey that are available for consumption by the adult predators. We employ a negative effect of this strategy in terms of diminished growth. We assume that the unavailable preys do not take part in the growth of prey population. Therefore, the total growth of the prey population at any instant of time is $r(1-n)x(1-x/k)$. We also assume that $\phi=0$ in the model (\ref{omodel}) , therefore, the modified system of equation becomes 
\begin{equation}\label{momodel}
\begin{cases}
    \frac{dx}{dt}=r(1-n)x(1-\frac{x}{k})-\alpha_1 x y- \alpha_2(1-n) x z\\
    \frac{dy}{dt}=\mu \alpha_2 (1-n)x z-\alpha_3 x y-\beta y-\gamma y\\
    \frac{dz}{dt}=\gamma y-m z
    \end{cases}
\end{equation}We discretize (\ref{momodel}) by Euler's forward method, and obtain
the discretized model 
\begin{equation}\label{eq3}
    \begin{cases}
    x(t+h)=x(t)+h\{r(1-n)x(t)(1-x(t)/k)-\alpha_1 x(t)y(t)-\alpha_2 (1-n) x(t)z(t)\}\\
y(t+h)=y(t)+h\{\mu \alpha_2(1-n) x(t)z(t) -\alpha_3 x(t)y(t)-\beta y(t)-\gamma y(t)\}\\
z(t+h)=z(t)+h\{\gamma y(t)-mz(t)\}    \end{cases}
\end{equation}
Considering $x(t+h)=x_{t+1}$, $y(t+h)=y_{t+1}$, $z(t+h)=z_{t+1}$, the system (\ref{eq3}) becomes
\begin{equation}\label{eq4}
    \begin{cases}
    x_{t+1}=x_t+h\{r(1-n)x_t(1-x_t/k)-\alpha_1 x_ty_{t}-\alpha_2 (1-n) x_tz_{t}\}\\
y_{t+1}=y_{t}+h\{\mu \alpha_2(1-n) x_tz_{t} -\alpha_3 x_ty_{t}-\beta y_{t}-\gamma y_{t}\}\\
z_{t+1}=z_{t}+h\{\gamma z_{t}-mz_{t}\}   
\end{cases}
\end{equation}
here, the variables $x_t$, $y_t$, and $z_t$ denote  the population sizes of prey, juvenile predator, and adult predator at generation t, where $t\in\mathbf{N}$. $r$ indicates the prey's growth rate, and $k$ represents the system's environmental carrying capacity. The predation rates of juvenile and adult predators are denoted by $\alpha_1$ and $\alpha_2$ respectively, $\mu$ is the conversion efficiency or reproduction rate of the adult predators, $\alpha_3$ is the prey counter-attacking rate to juvenile predators, $\beta$ is the juvenile predators' natural death rate, $\gamma$ is the juvenile predators' maturation rate, and $m$ is the adult predators' depletion rate in the absence of prey, and $n \in (0,1)$ is the coefficient of prey refuge.

\section{Equilibrium points and their stability} This section discusses the existence and stability of all biologically viable equilibrium points. After performing some calculations, all of the equilibrium points that are biologically feasible have been determined. These are vanishing equilibrium point $E_1(0,0,0)$, axial equilibrium point $E_2(k,0,0)$ and the coexisting equilibrium point $E_3=(x^*,y^*,z^*)$,where,
$$x^*=-\frac{m (\beta +\gamma )}{\alpha _3 m+\alpha _2 \gamma  \mu  (n-1)}$$

$$y^*=-\frac{m (n-1) r \left(m \left(\beta +\gamma +\alpha _3 k\right)+\alpha _2 \gamma  k \mu  (n-1)\right)}{k \left(\alpha _1 m-\alpha _2 \gamma  (n-1)\right) \left(\alpha _3 m+\alpha _2 \gamma  \mu  (n-1)\right)}$$,

$$z^*=-\frac{\gamma  (n-1) r \left(m \left(\beta +\gamma +\alpha _3 k\right)+\alpha _2 \gamma  k \mu  (n-1)\right)}{k \left(\alpha _1 m-\alpha _2 \gamma  (n-1)\right) \left(\alpha _3 m+\alpha _2 \gamma  \mu  (n-1)\right)}$$

\subsection{Vanishing Equilibrium($E_1$)} The vanishing equilibrium is $E_1(0,0,0)$. $E_1$ exists for all biologically possible parameter values. It is unstable in nature, as proven by the following theorem.

\begin{theorem}
The vanishing equilibrium $E_1$ is not stable.
\end{theorem}
\begin{proof}
The eigenvalues of the Jacobian matrix at $E_1(0,0,0)$ are given by $\lambda_1=1-h m$, $\lambda_2= h (r-n r)+1$, $\lambda_3=1-h (\beta +\gamma )$. It is obvious that $|\lambda_2|>1$ i.e., the equilibrium point $E_1$ is a saddle point. Morever, $|\lambda_1|>1$ and $|\lambda_3|>1$ i,e the fixed point $E_1$ is a source (repellor) if $h > \frac{2}{m}$ and $m < \gamma$. Hence proved. 
\end{proof}

\subsection{Axial equilibrium($E_2$):} The axial equilibrium point is given by $E_2(k,0,0)$. Clearly $E_2$ exists for all possible parameter values of the system (\ref{eq4}). The following theorem demonstrates that, under certain parametric conditions, the axial equilibrium point $E_2$ exhibits stability.

\begin{theorem}\label{th2}
 The axial equilibrium $E_2(k,0,0)$ is stable if and only if (\romannum{1}) $m < \frac{\beta + \gamma}{3}$, (\romannum{2}) $r < -\frac{2}{(-h + h n)}$, (\romannum{3}) $\beta < \gamma$,(\romannum{4}) $\mu < \frac{-m \beta - m \gamma - 
  k m \alpha_3}{-k \gamma \alpha_2 + 
  k n \gamma \alpha_2}$, and
  (\romannum{5}) $ h \le \frac{2}{m + \beta + \gamma + k \alpha_3}$.
\end{theorem}

\begin{proof}
    The Jacobian matrix of the model ((\ref{eq4})) at the axial equilibrium point $E_2(k,0,0)$ is 
    
    $$J_{e2}=\begin{bmatrix}
 1 + h (-1 + n) r & -h k\alpha_{1} & h k (-1 + n) \alpha_{2} \\
 0 & 1 - h (\beta + \gamma) - h k \alpha_{3} & -h k (-1 + n) \mu \alpha_{2} \\
0 & h \gamma & 1 - h m
\end{bmatrix}$$

  Now, the eigenvalues of the Jacobian matrix $J_{e2}$ are $\lambda_4=1 + h (-1 + n) r$, $\lambda_5=\frac{1}{2} (2 - h m - h \beta - h \gamma - h k\alpha_3 - h \sqrt{\theta} )$, and $\lambda_6=\frac{1}{2} (2 - h m - h \beta - h \gamma - h k\alpha_3 + h \sqrt{\theta})$. Here, $\theta=-4 k (-1 + n) \gamma \mu \alpha_2 + (-m + \beta + \gamma + k \alpha_3)^2$. The stability of the fixed point $E_2(k,0,0)$ is reliant on the absolute values of the eigenvalues of the Jacobian matrix evaluated at the axial equilibrium point $E_2$. The axial equilibrium $E_2(k,0,0)$ is stable if $|\lambda_4|<1$, $|\lambda_5|<1$ and $|\lambda_6|<1$,  which is possible when the conditions listed below are satisfied\\
    (\romannum{1}) $m < \frac{\beta + \gamma}{3}$, (\romannum{2}) $r < -\frac{2}{(-h + h n)}$, (\romannum{3}) $\beta < \gamma$,(\romannum{4}) $\mu < \frac{-m \beta - m \gamma - 
  k m \alpha_3}{-k \gamma \alpha_2 + 
  k n \gamma \alpha_2}$, and
  (\romannum{5}) $ h \le \frac{2}{m + \beta + \gamma + k \alpha_3}$. Hence, proved.
 
\end{proof}

\begin{figure}[H]
     \centering
     \begin{subfigure}[F]{0.45\textwidth}
         \centering
         \includegraphics[width=\textwidth]{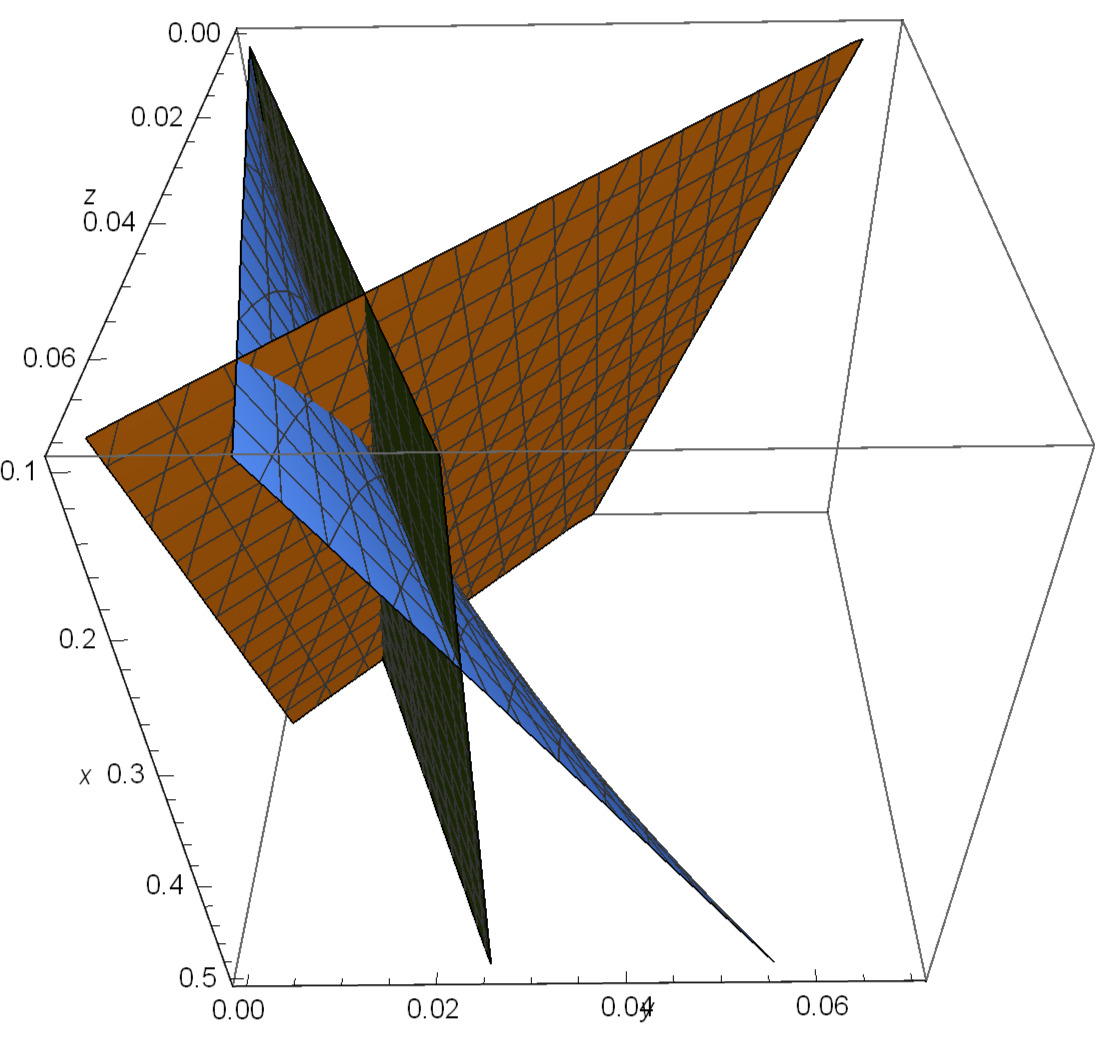}
         \caption{\emph{ Existence of the coexisting fixed point $E_3$,}}
         \label{1a}
     \end{subfigure}
     \hfill
     \begin{subfigure}[F]{0.45\textwidth}
         \centering
         \includegraphics[width=\textwidth]{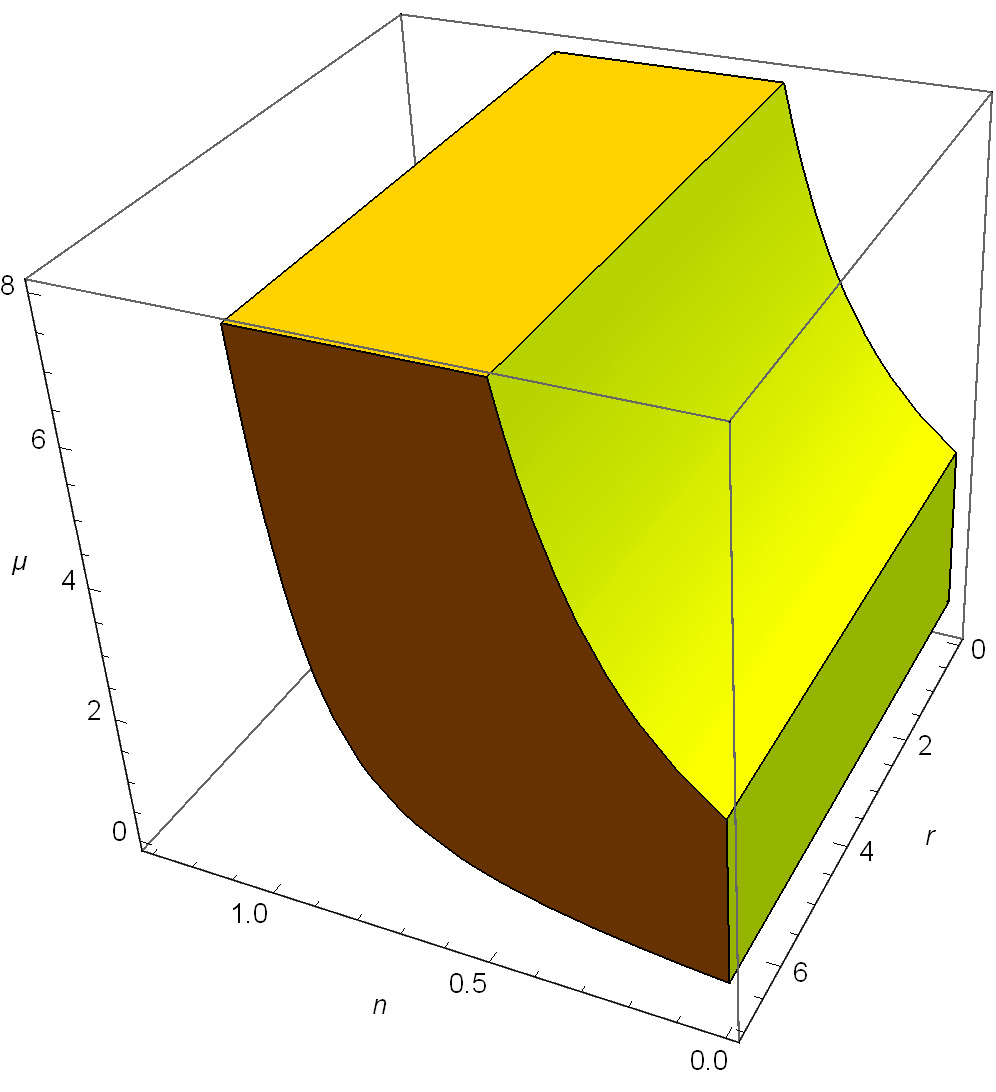}
         \caption{\emph{ Stability region of the coexisting fixed point $E_3$ in $nr\mu$-space}}
         \label{1b}
     \end{subfigure}
        \caption{\emph{The existence and stability of the coexisting fixed point $E_3$ are displayed using the parameter values stated in table (\ref{t1}).}}
        \label{figure0}
        \end{figure}
\subsection{Coexisting equilibrium($E_3$):} The coexisting equilibrium $E_3(x^*,y^*,z^*)$ exists if the following conditions are satisfied  \\\\
(\romannum{1}) $\gamma > -\frac{m \alpha_3}{m - \mu \alpha_2 + n \mu \alpha_2}$, (\romannum{2}) $\mu > -\frac{m}{-\alpha_2 + n \alpha_2}$, (\romannum{3}) $\beta < \frac{-m \gamma + \gamma \mu \alpha_2 - 
  n \gamma \mu \alpha_2 - m \alpha_3}{m}$, and (\romannum{4}) $0<n<1$. \\

The existence of the coexisting fixed point $E_3$ is evident from figure (\ref{1a}). The x-nullclines, y-nullclines, and z-nullclines of the model (\ref{eq4}) are depicted in figure (\ref{1a}), with the x-nullclines shown in brown, the y-nullclines in blue, and the z-nullclines in red. The parameter values used for this illustration are provided in table (\ref{t1}). The coexisting fixed point $E_3(0.232013, 0.0154654, 0.0419775)$ represents the intersection point of these nullclines. The Jacobian matrix of the system (\ref{eq4}) at any point $(x,y,z)$ is given by\\

$$J=\begin{bmatrix}
1 + h(-\frac{-((1 - n) r x)}{k}+(1 - n) r (1 - \frac{x}{k})-y\alpha_{1} - (1 - n)z\alpha_{2}) &-(1 + \gamma) r_2 & -h (1 - n) x \alpha_{2}]\\
h ((1 - n) z \mu \alpha_{2} - y\alpha_{3}) & 1 + h (-\beta - \gamma - x \alpha_{3}) & h (1 - n) x \mu \alpha_{2} \\
0 & h \gamma & 1 - h m
\end{bmatrix}$$
\\

Now, the Jacobian matrix of the system (\ref{eq4}) at the interior equilibrium point $E_3=(x^*,y^*,z^*)$ is\\

$$J_{e3}=\begin{bmatrix}
\frac{k (-1 + n) \gamma \mu \alpha_{2} + 
 m (-h (-1 + n) r (\beta + \gamma) + k \alpha_{3})}{k ((-1 + n) \gamma \mu \alpha_{2} + 
   m \alpha_{3})} & \frac{h m (\beta + \gamma) \alpha_1}{(-1 + n) \gamma \mu \alpha_{2} + 
 m \alpha_{3}} & \frac{h m (1 - n) (\beta + \gamma)\alpha_{2}}{(-1 + n) \gamma \mu \alpha_{2} + 
 m \alpha_{3}} \\
 \frac{h (-1 + n) r (k (-1 + n) \gamma \mu \alpha_{2} + 
   m (\beta + \gamma + k \alpha_{3}))}{k (m \alpha_{1} - (-1 + n) \gamma \alpha_{2})} & 1- \frac{h (-1 + n) \gamma (\beta + \gamma) \mu \alpha_{2}}{(-1 + n) \gamma \mu \alpha_{2} + 
 m \alpha_{3}} & \frac{h m (-1 + n) (\beta + \gamma) \mu \alpha_2}{(-1 + n) \gamma \mu \alpha_{2} + 
 m \alpha_{3}}\\
0 & h \gamma & 1 - h m
\end{bmatrix}$$
\\

The characteristic equation of the matrix $J_{e3}$ is as follows:
\begin{equation}\label{eq5}
\lambda^3+p_1\lambda^2+p_2\lambda+p_3=0
\end{equation}
where, \\\\
$p_1=\frac{k (h (\beta +\gamma +m+(n-1) r)-3)+h k (\alpha _2 (z^*-n z^*)+\alpha _3 x^*+\alpha _1 y^*)-2 h (n-1) r x^*}{k}$\\
$p_2=\frac{k (h^2 (m (\beta +\gamma +(n-1) r)+(n-1) r (\beta +\gamma ))-2 h (\beta +\gamma +m+(n-1) r)+3)+\delta_1}{k}$\\
$p_3=\frac{(h m-1) (k (h (n-1) r-1)-2 h (n-1) r x^*) (h (\beta +\gamma )+\alpha _3 h x^*-1)+\delta_2}{k}$,\\
$\delta_1=h (\alpha _3 x^* (k (h m+h (n-1) r-2)-\alpha _2 h k (n-1) z^*-2 h (n-1) r x^*)-\alpha _2 k (n-1) (z^* (h (\beta +\gamma +m)-2)
+\gamma  (-h) \mu  x^*+\alpha _1 h \mu  x^* z^*)+\alpha _1 k y^* (h (\beta +\gamma +m)-2))-2 h (n-1) r x^* (h (\beta +\gamma +m)-2)$, and \\
$\delta_2=\alpha _2 h (n-1) (\alpha _3 h k x^* (-h m z^*+\gamma  h y^*+z^*)-k z^* (h m-1) (h (\beta +\gamma )-1)+\gamma  (-h) \mu  x^* (-h k (n-1) r+2 h (n-1) r x^*+k))+\alpha _1 h k (\alpha _2 h \mu  (n-1) x^* (-h m z^*+\gamma  h y^*+z^*)+y^* (h m-1) (h (\beta +\gamma )-1))$.\\

The subsequent theorem provides proof for the stability of the fixed point, $E_3=(x^*,y^*,z^*)$.

\begin{theorem}\label{th3}
The coexisting equilibrium $E_3$ is locally stable if and only if $|p_1+p_3|<1+p_2$, $p_2-p_1 p_3<1-p_3^2$, and $|p_1-3p_3|< 3-p_2$.
\end{theorem}
\begin{proof}
Please refer to  Theorem 3.2 in \cite{ali}.
\end{proof}

\section{Neimark-Sacker bifurcation} The Neimark-Sacker bifurcation is a well-known bifurcation phenomenon that occurs in dynamical systems when a stable limit cycle experiences a loss of stability, leading to the emergence of an invariant torus or periodic cycles. In order to analyse the Neimark-Sacker bifurcation phenomenon concerning the coexisting equilibrium $E_3$, it is necessary to utilise the explicit criterion \cite{wen} outlined below.

\begin{theorem}\label{t4} \cite{wen}
Given a discrete dynamical system of l dimensions: $Z_{u+1}=f_{v}(Z_{u})$, where $v \in R$ denotes a bifurcation parameter. Assume that $ Z^{*}$ is a fixed point of $f_{v}$. Then, the
characteristic polynomical for Jacobian matrix $J(Z^{*}) = (a_{ij})_{l \times l}$ of l-dimensional
map $f_{v}$ is as follows:
$$P_{v}(\lambda)=\lambda^{l}+b_1\lambda^{l-1}+b_2\lambda^{l-2}+b_{l-1}\lambda+b_{l}$$
where, $b_{i}=b_{i}(v,c)$,i=1,2,3,...l, $c$ represents either a control parameter or another parameter that requires determination. Let us consider, a sequence of determinants of the type $(\Delta_{i}^{\pm}(v,c))_{i=0}^{l}$ such that $\Delta_{0}^{\pm}(v,c)=1$, and $\Delta_{i}^{\pm}(v,c)=det(K_1 \pm K_2)$,where \\

$K_1=\begin{bmatrix}
1 & b_1 & b_{2}&.....&b_{l-1}\\
0 & 1 &b_{1}&.....& b_{l-2} \\
0 & 0 &1&.....& b_{l-3} \\
...&...&...& ...&... \\
0 & 0 & 0 &...& 1
\end{bmatrix}$, $K_2=\begin{bmatrix}
b_{l-i+1} & b_{l-i+2} &b_{l-1}&.....&b_{l}\\
b_{l-i+2} & b_{l-i+3} &b_{l}&.....& 0 \\
...& ... &...& ... \\
b_{l-1}&b_{l}&...&0&0 \\
b_{l} & 0 & ... & 0&0
\end{bmatrix}$\\\\
Furthermore, it is assumed that the subsequent criteria are true:\\\\
\textbf{1st criteria:} Eigenvalue requirement: $\Delta_{l-1}^{-}(v_{0},c)=0$, $\Delta_{l-1}^{+}(v_{0},c)>0$, $P_{v_{0}}(1)>0$, $(-1)^{l}P_{v_{0}}(-1)>0$, $\Delta_{i}^{\pm}(v_{0},c)>0$, for i=l-3,l-5,...,2 (or 1), when l is odd or even, respectively.\\\\
\textbf{2nd criteria:} Transversality requirement: $\frac{d}{dv}(\Delta_{l-1}^{-}(v,c))_{v=v_{0}}\neq 0$.\\\\
\textbf{3rd criteria:} Non-resonance condition: $cos(\frac{2\pi}{j})\neq \Psi,$ or resonance condition $cos(\frac{2\pi}{j}) = \Psi$, where j= 3, 4 ,5,.... and $\Psi=1-(0.5P_{v_{0}}(1)\Delta_{l-3}^{-}(v_{0},c)/\Delta_{l-2}^{+}(v_{0},c))$; then, at a critical point $v_{0}$ , Neimark-Sacker bifurcation takes place.
\end{theorem}

The following theorem offers criteria that establish the occurrence of Neimark-Sacker bifurcation for system (\ref{eq4}) with respect to the bifurcation parameter $r$.

\begin{theorem}\label{t5}
    The fixed point $E_3$ experiences Neimark-Sacker bifurcation at the critical value $r=r^{ns}$ depends on the satisfaction of the specified conditions.\\
(\romannum{1}) $1-p_2+ p_3(p_1-p_3) = 0$\\
(\romannum{2}) $1+p_2-p_3(p_1+p_3) > 0$\\
(\romannum{3}) $1+p_1+p_2+p_3 > 0$\\
(\romannum{4}) $1-p_1+p_2-p_3 > 0$\\
(\romannum{5}) $\frac{d}{dr}(1-p_2+ p_3(p_1-p_3))_{r=r^{ns}}\neq 0$\\
(\romannum{6}) $cos(\frac{2\pi}{j})\neq 1-\frac{1+p_1+p_2+p_3}{2(1+p_3)},$ j=3, 4, 5,...,\\
where, $p_1$, $p_2$ represents the coefficients of 
 $\lambda^2$, $\lambda$ , and $p_3$ represents the constant term in the equation (\ref{eq5}). $r^{ns}$ is a real root of the equation $1-p_2+ p_3(p_1-p_3) = 0$.
\end{theorem}

\begin{proof} Let us consider, r as a bifurcation parameter and  $l=3$. Now, following the theorem (\ref{t4}) and using the equation (\ref{eq5}), we compute the following values \\
$$\Delta_{2}^{-}(r)=1-p_2+ p_3(p_1-p_3) = 0,$$
$$\Delta_{2}^{+}(r)=1+p_2-p_3(p_1+p_3) > 0,$$
$$P_{r^{ns}}(1)=1+p_1+p_2+p_3 > 0,$$
$$(-1)^{3}P_{r^{ns}}(-1)=1-p_1+p_2-p_3 > 0,$$
$$\frac{d}{dr}(\Delta_{2}^{-}(r))_{r=r^{ns}}\neq 0,$$
$$1-\frac{(0.5P_{v_{0}}(1)\Delta_{0}^{-}(r))}{\Delta_{1}^{+}(r))}=1-\frac{1+p_1+p_2+p_3}{2(1+p_3)}$$.
\end{proof}

Other parameters can also be taken into consideration as the bifurcation parameter, leading to similar results.

\section{Period-doubling bifurcation} The period doubling bifurcation is a notable occurrence in discrete dynamical systems, wherein the system experiences a series of bifurcations that lead to the doubling of the period of its orbits. To conduct an analysis of the Period-doubling bifurcation for the map (\ref{eq4}) about the fixed point $E_3$, a specific criteria \cite{wen2} is required, as described in the following section.
\begin{theorem}\label{t6}
    Given a discrete dynamical system of l dimensions: $Z_{u+1}=f_{r}(Z_{u})$, where $r \in R$ denotes a bifurcation parameter. Assume that $ Z^{*}$ is a fixed point of $f_{r}$. Then, the
characteristic polynomical for Jacobian matrix $J(Z^{*}) = (a_{ij})_{l \times l}$ of l-dimensional
map $f_{r}$ is as follows:
$$P_{r}(\lambda)=\lambda^{l}+b_1\lambda^{l-1}+b_2\lambda^{l-2}+b_{l-1}\lambda+b_{l}$$
where, $b_{i}=b_{i}(r)$, i=1,2,3,...l . Let us consider, a sequence of determinants of the type $(\Delta_{i}^{\pm}(r))_{i=0}^{l}$ such that $\Delta_{0}^{\pm}(r)=1$, and $\Delta_{i}^{\pm}(r)=det(K_1 \pm K_2)$, where $K_1$ and $K_2$ are same as given in theorem (\ref{eq4}). Then, a period-doubling bifurcation occurs at a critical value $r = r^{pb}$ if and only if the following requirements are fulfilled.\\
(\romannum{1}) Eigenvalue requirement: $P_{r^{pb}}(-1) = 0$, $P_{r^{pb}}(1) > 0$ , $\Delta_{l-1}^{\pm}(r^{pb})>0$, $\Delta_{i}^{\pm}(r^{pb})>0$, i=l-2, l-4,...., 2(or 1), when n is even or odd, repectively.\\
(\romannum{2}) Transversality requirement: $\frac{\sum_{i=1}^{l} b_{i}^{'}(-1)^{l-i}}{\sum_{i=1}^{l} (l-i+1)(-1)^{l-i}c_{i-1}}\neq 0$; where $b_{i}^{'}$ represents the first derivative of $b_{i}$ with respect to $r$ at $r=r^{pb}$.
\end{theorem}
By employing the aforementioned theorem, we determine the conditions that lead to the occurrence of period-doubling bifurcation in relation to the parameter $r$.
\begin{theorem}\label{t7} \cite{ali}
    The  fixed point $E_3$ of the map (\ref{eq4}) exhibits a period-doubling bifurcation at $r = r^{pb}$ when the subsequent specified conditions are satisfied.\\
(\romannum{1}) $1-p_2+ p_3(p_1-p_3) > 0,$\\
(\romannum{2}) $1+p_2-p_3(p_1+p_3) > 0,$\\
(\romannum{3}) $1\pm p_2 > 0,$\\
(\romannum{4}) $1+p_1+p_2+p_3 > 0,$ and \\
(\romannum{5}) $-1+p_1-p_2+p_3 = 0$,\\
where, the values of $p_1$, $p_2$, and $p_3$ are provided in equation (\ref{eq5}).
\end{theorem}
\begin{figure}[H]
     \centering
     \begin{subfigure}[F]{0.4\textwidth}
         \centering
         \includegraphics[width=1.2\textwidth]{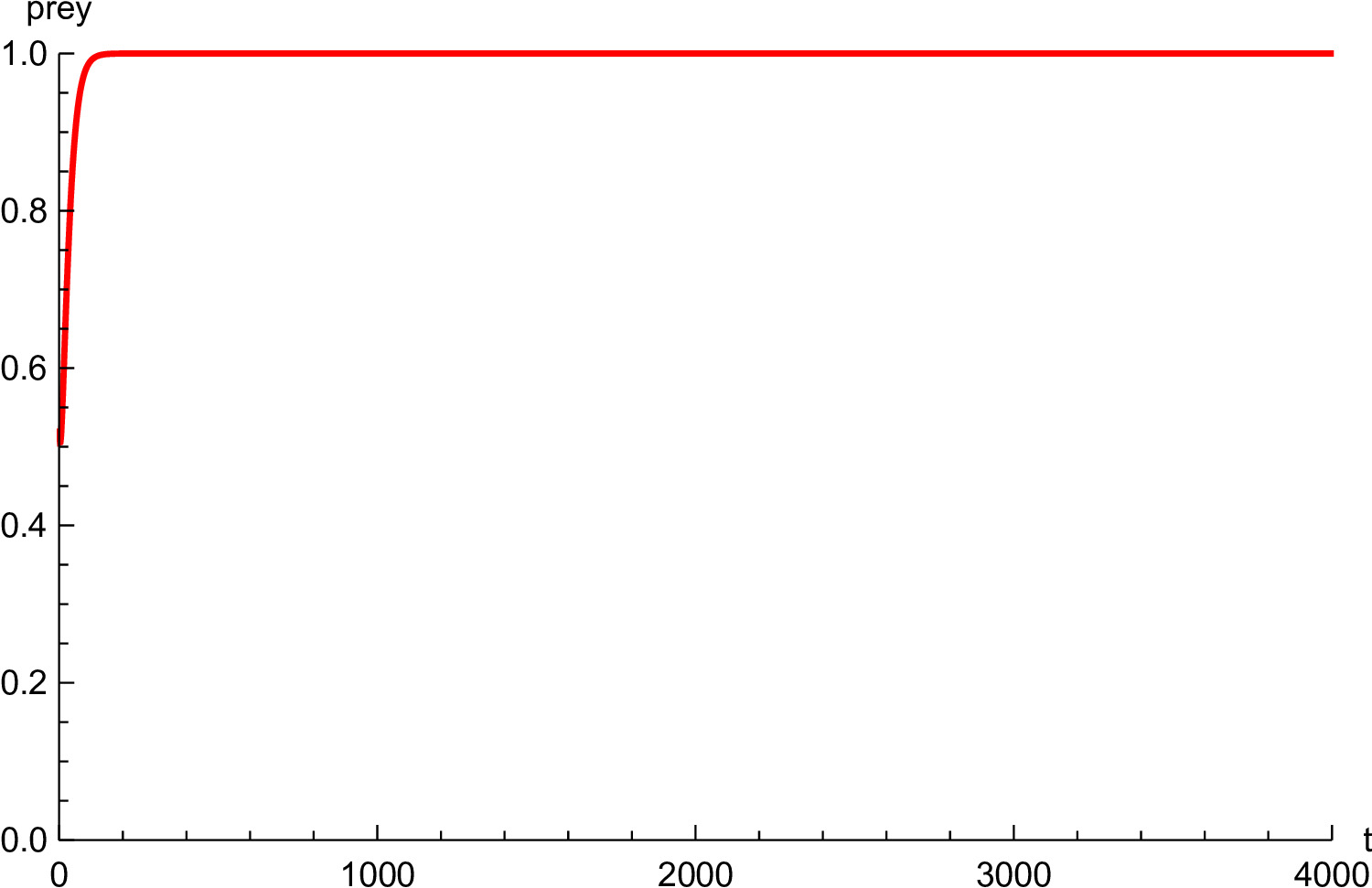}
         \caption{\emph{Time series of prey }}
         \label{fig:y equals x}
     \end{subfigure}
     \hfill
     \begin{subfigure}[F]{0.5\textwidth}
         \centering
         \includegraphics[width=\textwidth]{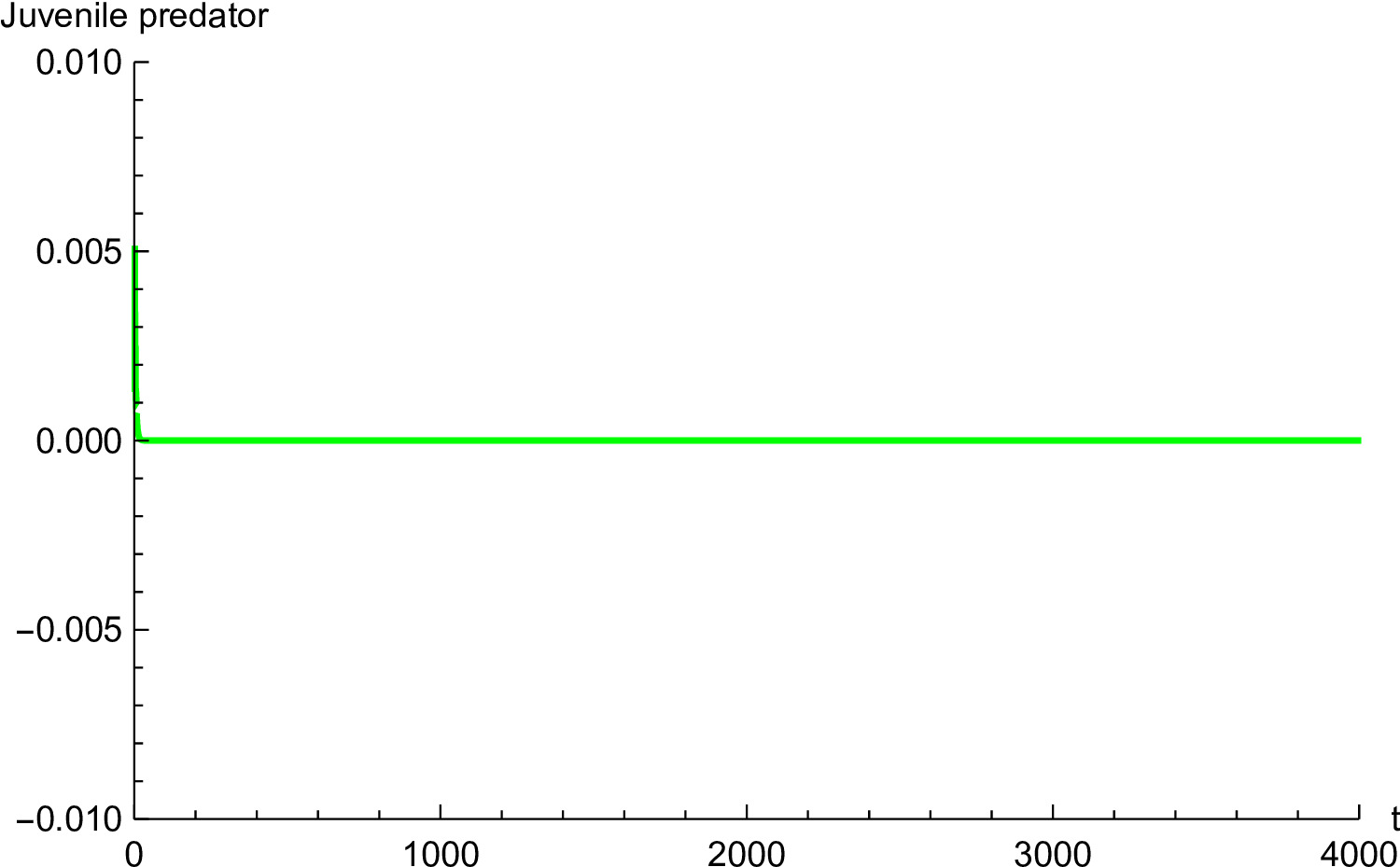}
         \caption{\emph{Time series of juvenile predator}}
     \end{subfigure}

     \begin{subfigure}[F]{0.5\textwidth}
         \centering
         \includegraphics[width=\textwidth]{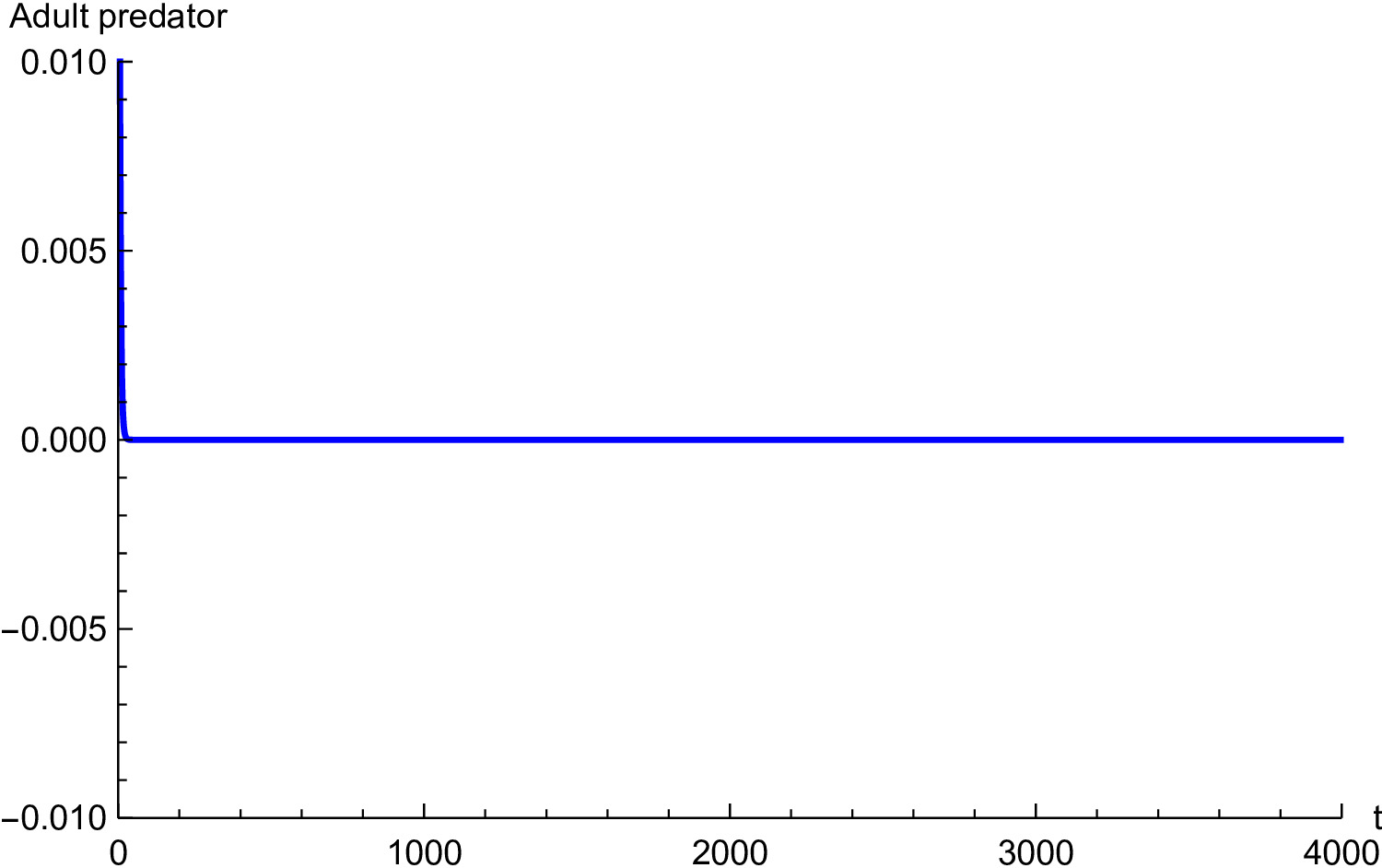}
         \caption{\emph{Time series of adult predator}}
     \end{subfigure}
        \caption{\emph{The stability of the axial fixed point $E_2$ is portrayed using parameter values $h=0.1$, $r=0.5$, $n=0.01$, $\alpha_1=10$, $\alpha_2=10$, $\mu=0.275944$, $\alpha_3=0.03$, $\beta=5.3$, $\gamma=9.5$, $k=100$, and $m=3.5$}}

        \label{fig1}
\end{figure}

\section{Numerical simulation}
In this section, we present numerical simulations to validate the theoretical findings previously discussed in the preceding sections. The hypothetical parameter values depicted in table (\ref{t1}) are taken into consideration. The Mathematica software is employed for conducting numerical simulations to facilitate the analysis of the obtained results.

At first we consider the parameter values $h=0.1$, $r=0.5$, $n=0.01$, $\alpha_1=10$, $\alpha_2=10$, $\mu=0.275944$, $\alpha_3=0.03$, $\beta=5.3$, $\gamma=9.5$, $k=100$, and $m=3.5$. In order to validate the stability requirements of the axial equilibrium point $E_2$ as given in theorem (\ref{th2}), these parameter values are employed. By utilising the given parameter values, the eigenvalues of the Jacobian matrix $J_{e2}$ can be determined. These eigenvalues are $|\lambda_4|=0.9505<1$, $|\lambda_5|=0.678371<1$, and $|\lambda_5|=0.845371<1$. As a result, as illustrated in the figure (\ref{fig1}), the fixed point $E_2$ is stable. We now take the parameter values listed in the following table (\ref{t1}) to validate the stability criteria of the coexisting equilibrium $E_3$ as mentioned in the theorem (\ref{th3}).
\begin{table}[H]
\begin{center}
$\begin{array}{|c|c|}
\hline
 \text{parameter} & \text{values} \\
\hline
 h & 0.1 \\
\hline
 k & 1 \\
\hline
 m & 3.5 \\
\hline
 n & 0.01 \\
\hline
 r & 0.75 \\
\hline
 \beta  & 5.3 \\
\hline
 \gamma  & 9.5 \\
\hline
 \mu  & 2.375 \\
\hline
 \alpha_1 & 10. \\
\hline
 \alpha_2 & 10. \\
\hline
 \alpha_3 & 0.03 \\
\hline
\end{array}$
\end{center}
\caption{Parameter values of the system (\ref{eq4}) for the purpose of numerical simulation}
\label{t1}
\end{table}

Using these parameter values, we compute the characteristic equation of the Jacobian matrix $J_{e_3}$ which is given by
\begin{equation}\label{eq6}
\lambda^3-1.15208\lambda^2-0.64142\lambda+0.823035=0
\end{equation}
    Comparing equation (\ref{eq6}) with equation (\ref{eq5}), we have $p_1=-1.15208$, $p_2=-0.64142$ and $p_3=0.823035$. Npw, we have $1+p_2-|p_1+p_3|=0.0295379>0$, $1-p_3^2-p_2+p_1 p_3=0.0158339>0$, and $ 3-p_2-|p_1-3p_3|=0.0202382>0$. Hence, it can be concluded that the fixed point $E_3$ exhibits stability in accordance with theorem (\ref{th3}). It is readily apparent in figures  (\ref{1b}) and (\ref{fig2}). The stability region of the coexisting fixed point in the $nr\mu$-space is depicted in Figure (\ref{1b}), taking into account the parameter values provided in table (\ref{t1}), with the exception of the parameters n, r, and $\mu$.
\begin{figure}[H]
     \centering
     \begin{subfigure}[F]{0.4\textwidth}
         \centering
         \includegraphics[width=\textwidth]{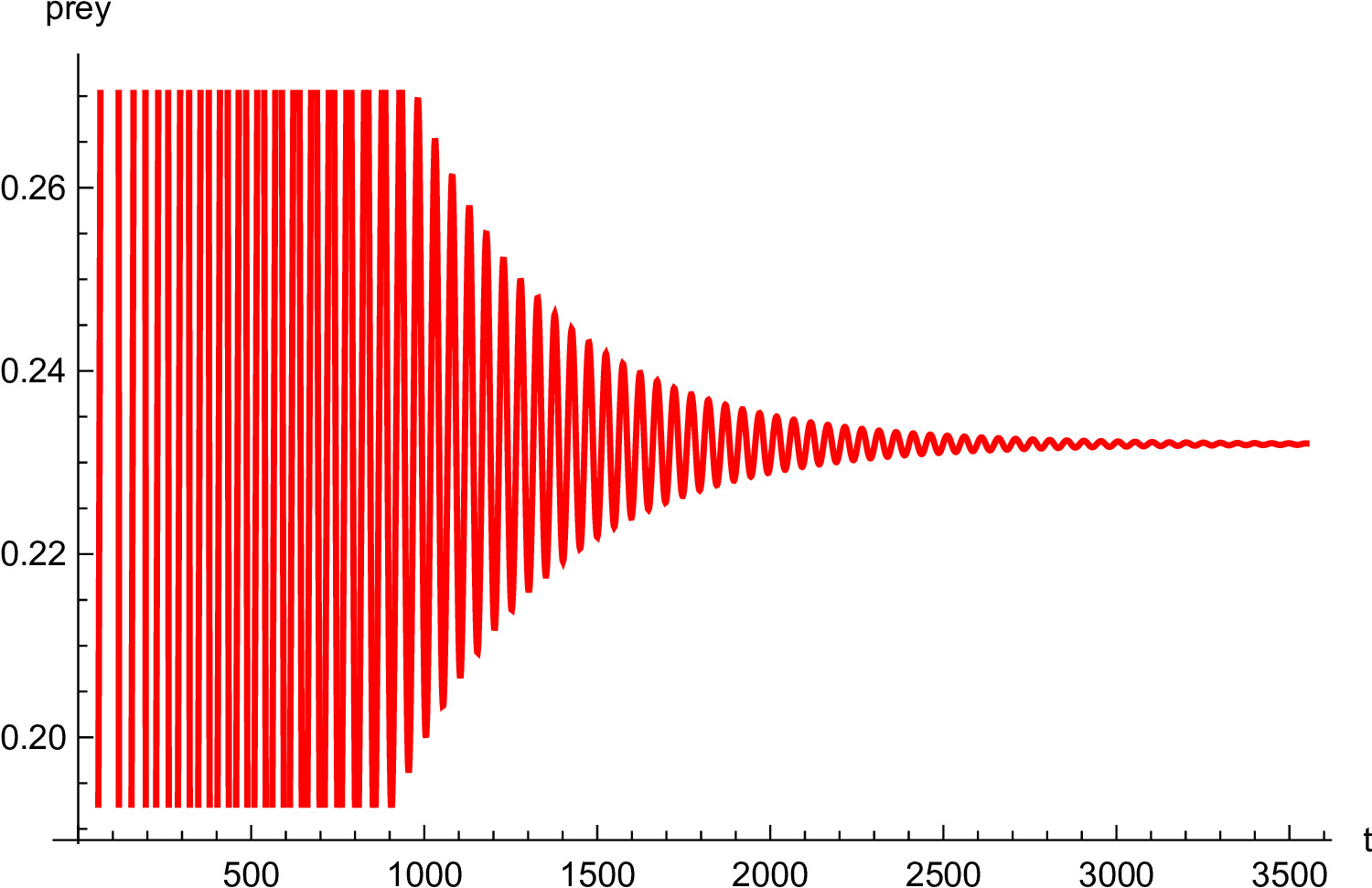}
         \caption{\emph{Time series of Prey }}
         \label{fig:y equals x}
     \end{subfigure}
     \hfill
     \begin{subfigure}[F]{0.4\textwidth}
         \centering
         \includegraphics[width=\textwidth]{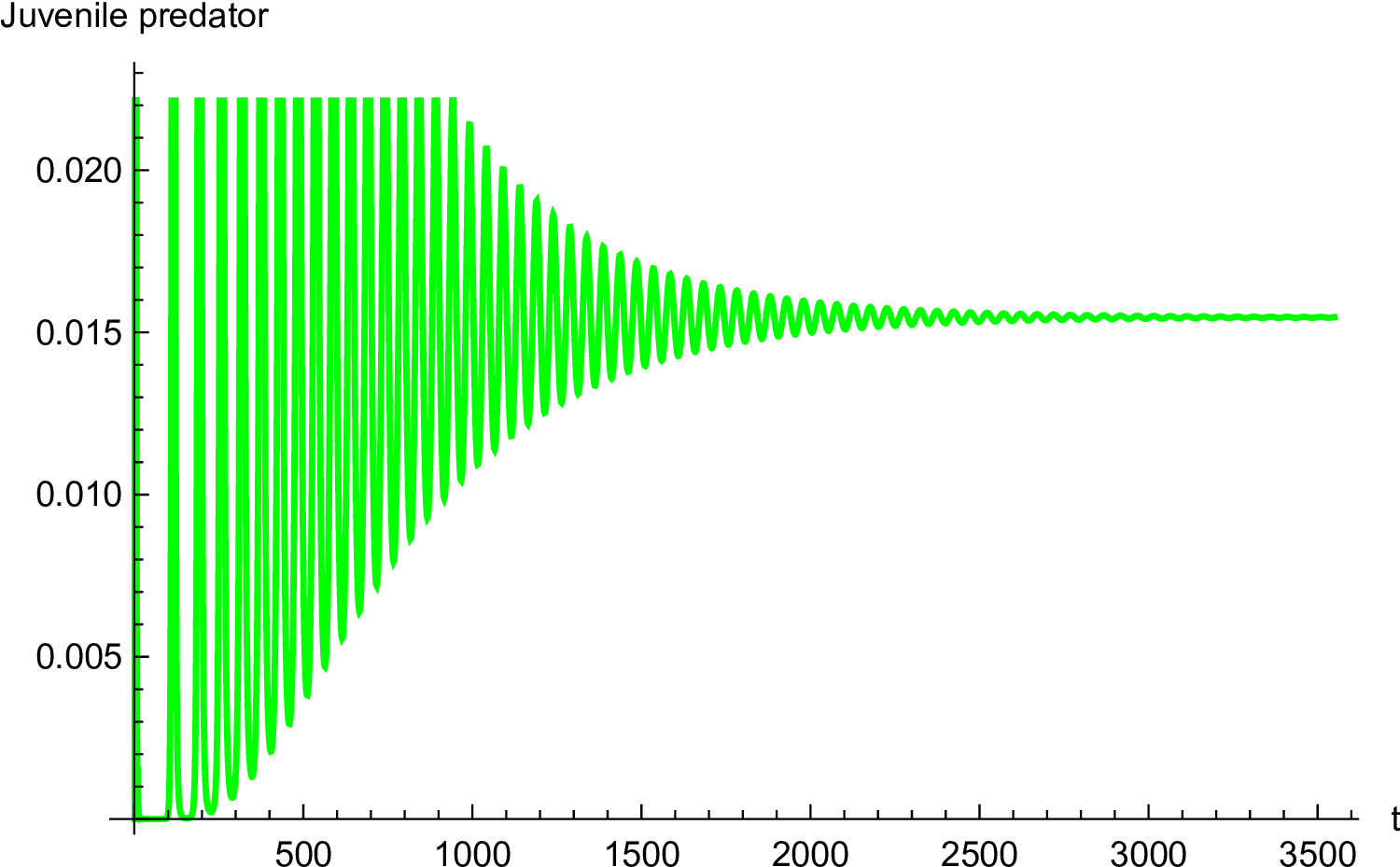}
         \caption{\emph{Time series of juvenile predator}}
     \end{subfigure}

     \begin{subfigure}[F]{0.4\textwidth}
         \centering
         \includegraphics[width=\textwidth]{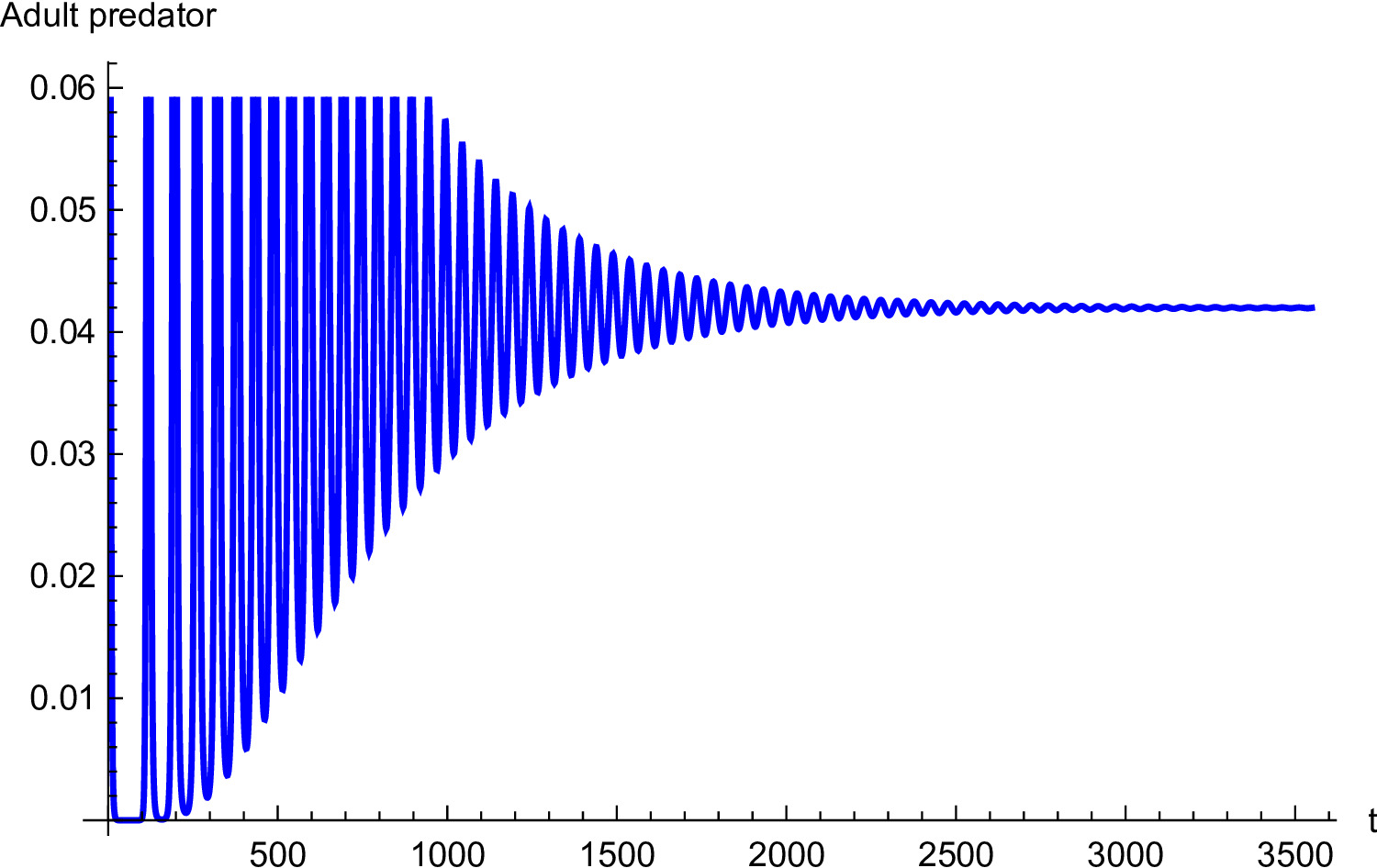}
         \caption{\emph{Time series of adult predator}}
     \end{subfigure}
        \caption{\emph{The stability of the coexisting fixed point $E_3$ is demonstrated using parameter values from table (\ref{t1})}}

        \label{fig2}
\end{figure}

In order to validate the outcome presented in theorem (\ref{t5}), we examine the parameter values $r\in (0.2,0.7)$ , $\mu=3.07227$, while keeping the remaining parameters consistent with those specified in table (\ref{t1}). Parameter $r$ is used as the bifurcation parameter in this case. In the vicinity of the parameter value $r=0.539=r^{ns}$, the fixed point (0.179337, 0.0118774, 0.0322387) undergoes a transition in stability, transitioning from a stable fixed population to a stable periodic population as a result of a Neimark-Sacker bifurcation. For a given value of $r=0.539$, $\mu=3.07227$, and assuming all other parameters are as specified in table (\ref{t1}), we find the characteristic equation of the Jacobian matrix at the fixed point $E_3$
\begin{equation}\label{eq7}
\lambda^3-1.15989\lambda^2-0.645119\lambda+0.827696=0
\end{equation}

here, $p_1=-1.15989$, $p_2=-0.645119$, and $p_3=0.827696$.\\

Now, we find
$1-p_2+ p_3(p_1-p_3) = 0$ , $1+p_2-p_3(p_1+p_3)=0.629838 > 0$, $1+p_1+p_2+p_3=0.0226851 > 0$, ,$1-p_1+p_2-p_3=0.687076 > 0$, $\frac{d}{dr}(1-p_2+ p_3(p_1-p_3))_{r=r^{ns}}=0.0000654356 \neq 0$, and using the equation $cos(\frac{2\pi}{j})=0.993794$, one obtains $j=\pm 56.3685$, therefore, the non-resonance criterion is also satisfied
   i.e, all the necessary conditions for the occurrence of the Neimark-Sacker bifurcation have been satisfied, as stated in theorem (\ref{t5}). The visual representations for the same can be observed in the diagrams depicted in Figure (\ref{fig3}). Furthermore, by selecting a value of $r$ that is less than $r^{ns}$, specifically $r=0.48$, and $\mu=3.07227$, while keeping all other parameters as specified in table (\ref{t1}), it is found that all the eigenvalues of the Jacobian matrix $J_{e3}$ are $\lambda_7=0.994474 + 0.104977 i$, $\lambda_8=0.994474 - 0.104977 i$, and $\lambda_9=-0.828008$ i.e., $|\lambda_7|<1$, $|\lambda_8|<1$, and $|\lambda_9|<1$. This confirms the stability of the fixed point $E_3$. Although, with a value of $r=0.6>r^{ns}$ and all other parameter values remaining the same as previously stated, the eigenvalues of $J_{e3}$ are found to be $\lambda_7=0.998289 + 0.105108 i$, $\lambda_8=0.998289 - 0.105108 i$, and $\lambda_9=-0.828019$ i.e., $|\lambda_7|>1$, and $|\lambda_8|>1$. This observation confirms the unstable nature of the fixed point.

\begin{figure}[H]
     \centering
     \begin{subfigure}[F]{0.4\textwidth}
         \centering
         \includegraphics[width=1.2\textwidth]{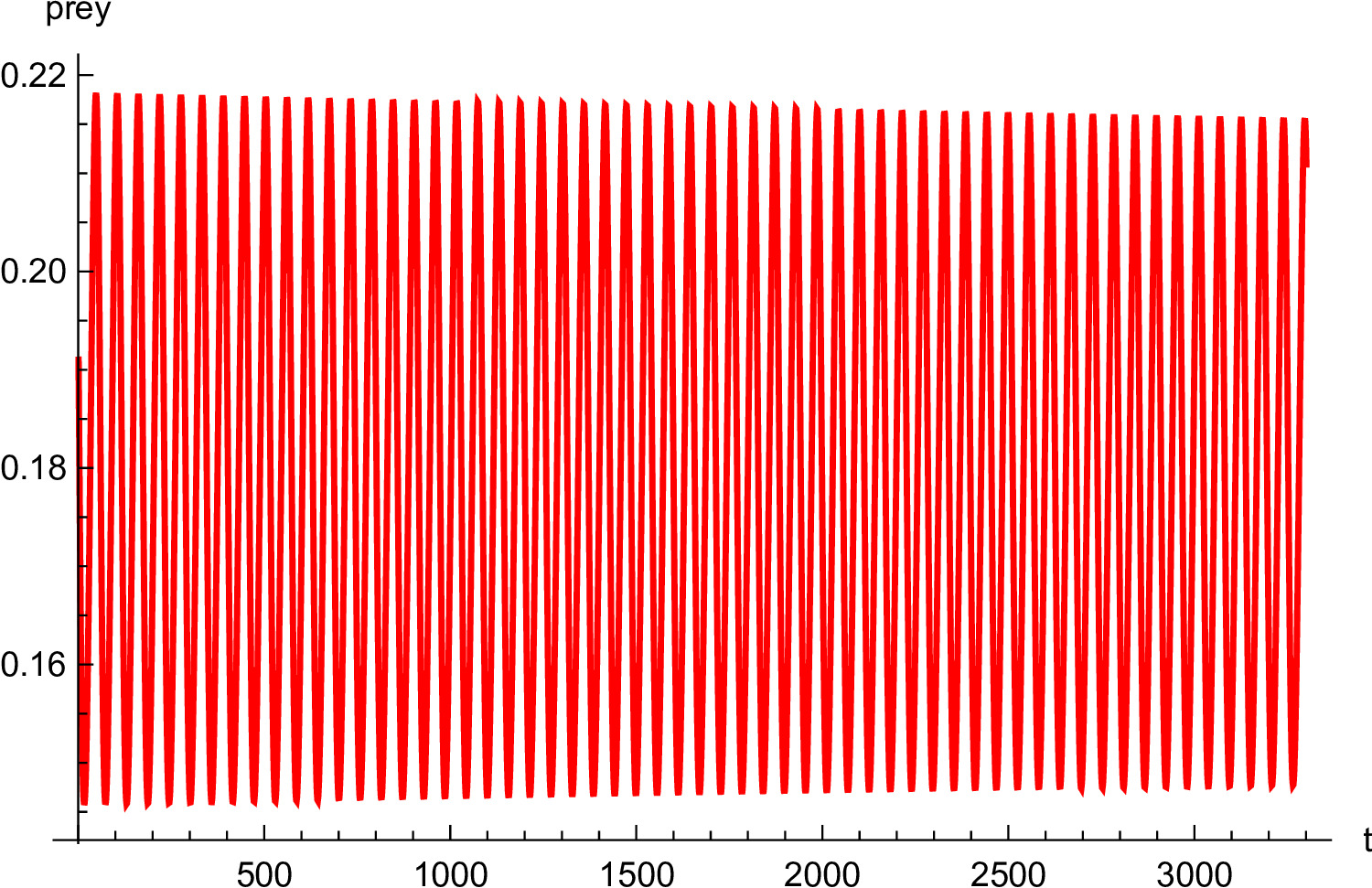}
         \caption{\emph{Time series of prey }}
         \label{fig:y equals x}
     \end{subfigure}
     \hfill
     \begin{subfigure}[F]{0.5\textwidth}
         \centering
         \includegraphics[width=\textwidth]{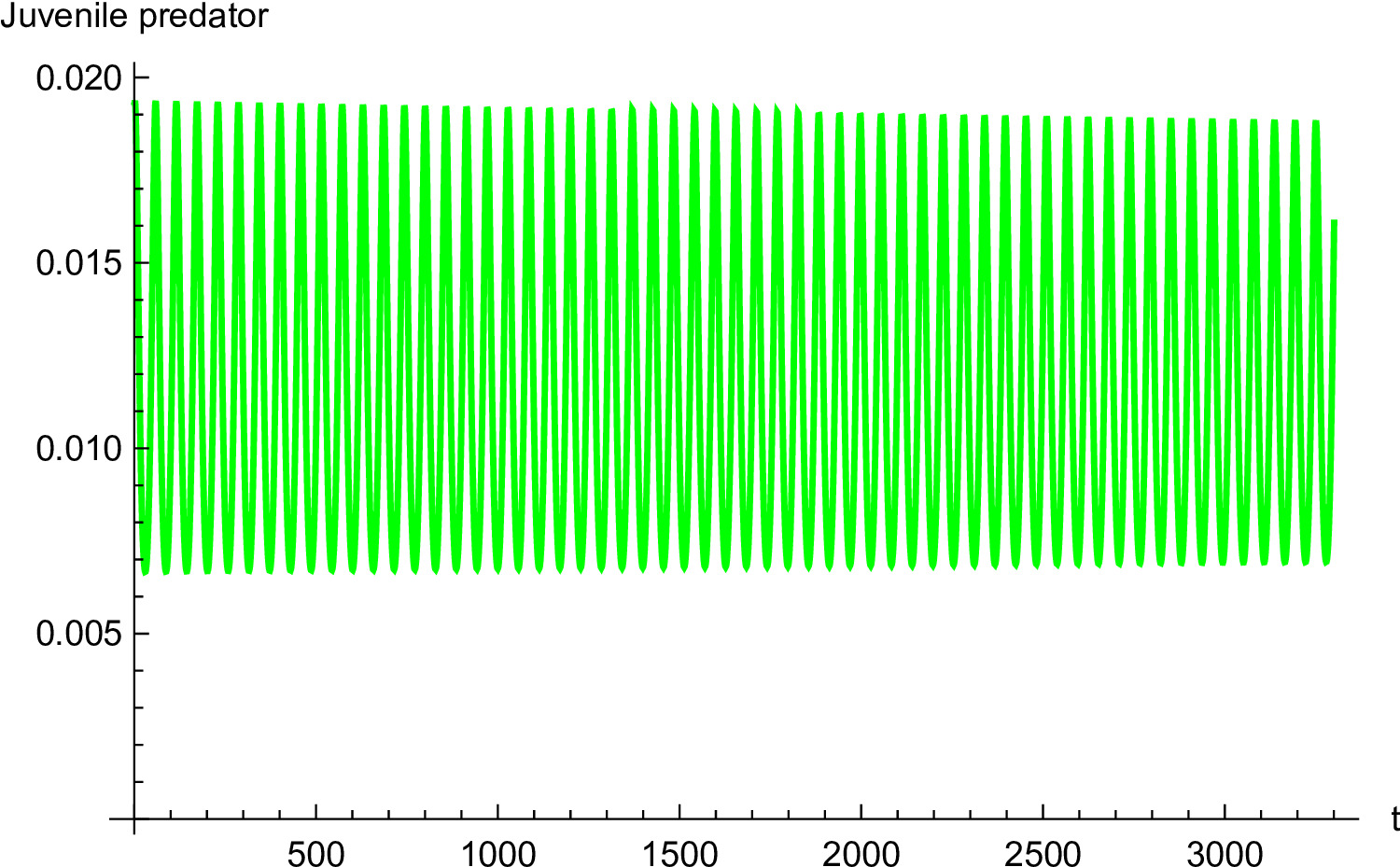}
         \caption{\emph{Time series of juvenile predator}}
     \end{subfigure}

     \begin{subfigure}[F]{0.45\textwidth}
         \centering
         \includegraphics[width=\textwidth]{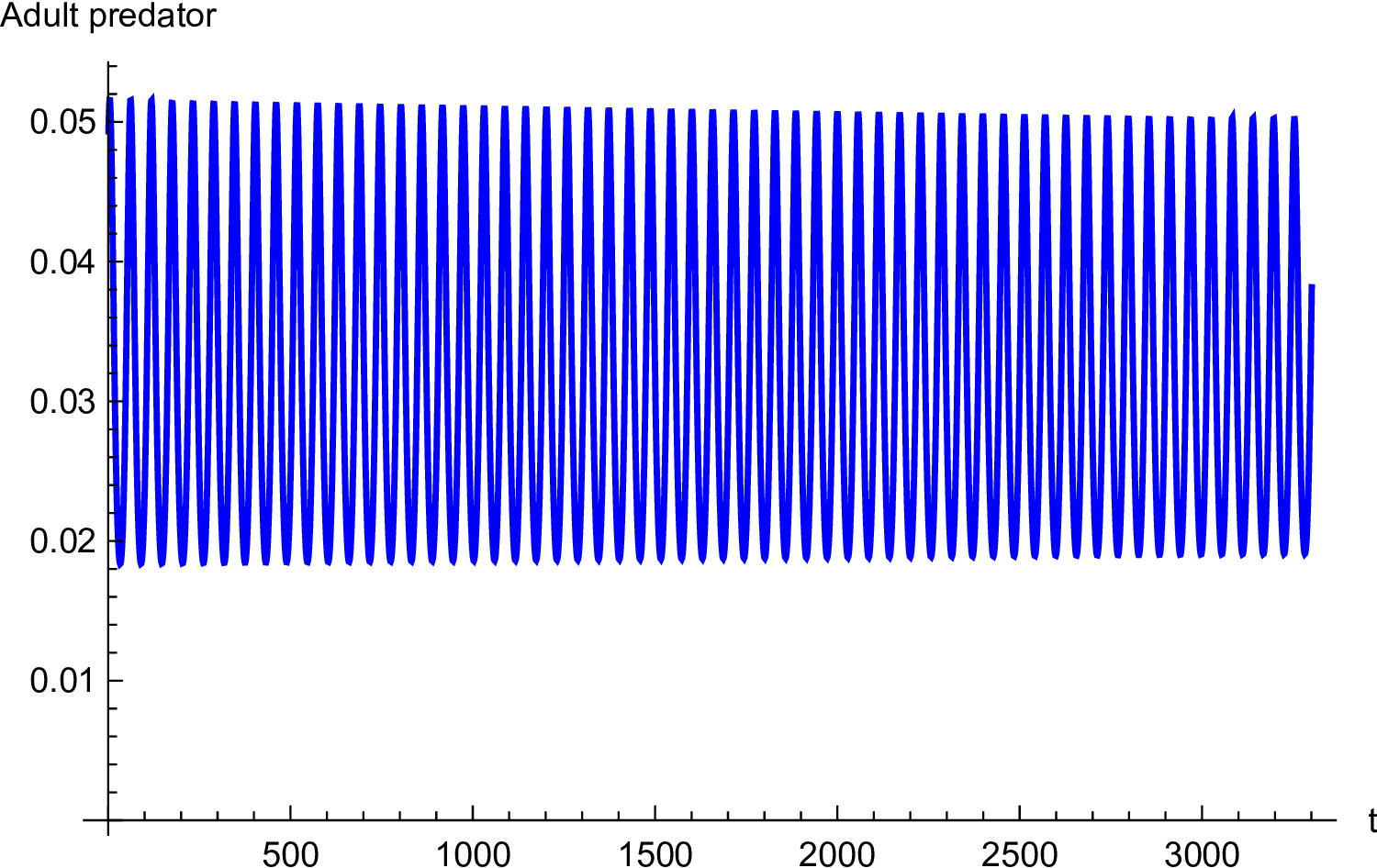}
         \caption{\emph{Time series of adult predator}}
     \end{subfigure} \hfill
      \begin{subfigure}[F]{0.45\textwidth}
         \centering
         \includegraphics[width=\textwidth]{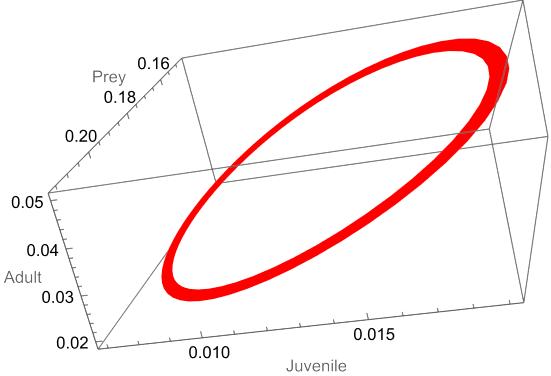}
         \caption{\emph{Phase portrait}}
     \end{subfigure}
        \caption{\emph{The occurrence of the Neimark-Sacker bifurcation is demonstrated when the parameter $r$ is varied. These figures have been generated utilising the parameter values $r=0.539$ and $\mu=3.07227$, while the remaining parameter values are sourced from table (\ref{t1}).}}

        \label{fig3}
        \end{figure}

 To examine the conditions for the occurrence of period-doubling bifurcation, as laid out in theorem (\ref{t7}), we consider the parameter values $r\in (22,25)$ , $\mu=0.59977$, and the remaining parameters are maintained in accordance with the values provided in the table (\ref{t1}). Here , r is taken as the bifurcation parameter. The stability of the coexisting fixed point (0.920016, 0.0509269, 0.13823) undergoes a transition from stability to instability at the value $r=23.7137=r^{pd}$, resulting from a period-doubling bifurcation. At the period-doubling bifurcation point $r=r^{pd}$, the fixed point $E_3$ undergoes destabilisation, resulting in the emergence of two points that constitute the period-2 solution. The characteristic polynomial of $J_{e3}$ with $r=r^{pd}$, $\mu=0.59977$ and the other parameters as stated previously, 
\begin{equation}\label{eq8}
\lambda^3+0.992639\lambda^2-0.951366\lambda-0.944005=0
\end{equation}
here, $p_1=0.992639$, $p_2=-0.951366$, and $p_3=-0.944005$.\\

Next, we proceed with the computation of the expression $1-p_2+ p_3(p_1-p_3)=0.123164> 0$, $1+p_2-p_3(p_1+p_3)=0.094544 > 0$, $1+ p_2=0.0486336 > 0$, $1- p_2=1.95137 > 0$, $1+p_1+p_2+p_3=0.0972673 > 0$, and 
 $-1+p_1-p_2+p_3 = 0$ which implies the fact that, as stated in theorem (\ref{t7}), all the requirements for a period-doubling bifurcation are met in the vicinity of the coexisting fixed point $(0.920016, 0.0509269, 0.13823)$ at the critical value of the bifurcation parameter $r = r^{pd}$. The figure (\ref{fig4}) illustrates the period-doubling bifurcation diagram associated with the parameter $r$. Furthermore, when $r = 22 < r^{pd}$ and all other parameter values remain unchanged as previously discussed, the eigenvalues of the matrix $J_{e3}$ are found to be $0.975264$, $-0.90591 + 0.124984i$, and $-0.90591 - 0.124984i$. These eigenvalues have modulus less than 1, indicating that the fixed point $E_3$ is stable. However, considering the value of $r$ is 25, which is greater than $r^{pd}$, and assuming that the remaining parameter values are the same as those discussed earlier, the eigenvalues of the Jacobian matrix $J_{e3}$ are determined to be $|-1.18461|>1$, $|0.975302|<1$, and $|-0.900491|<1$, thus confirming the unstable nature of the coexisting fixed point $E_3$.

        \begin{figure}[H]
            \centering
            \includegraphics[width=0.8\textwidth]{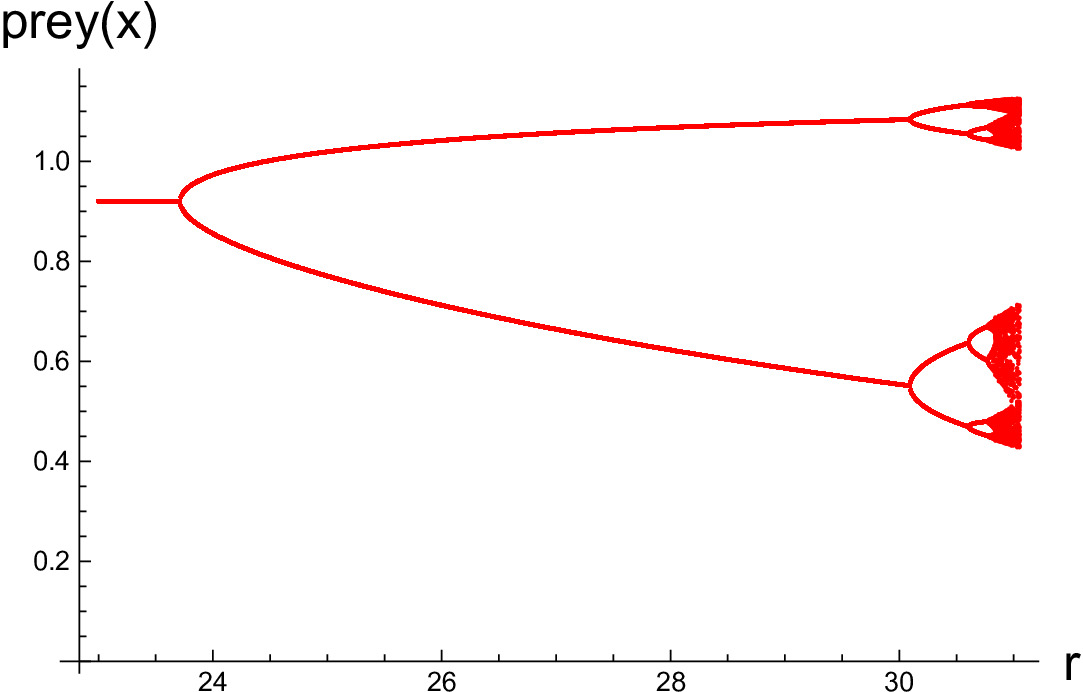}
            \caption{Period-doubling bifurcation is depicted in relation to the bifurcation parameter $r$. The figure is constructed using the parameters $r=23.7137$ and $\mu=0.59977$, with the remaining parameters are kept consistent with the values specified in table (\ref{t1}). }
            \label{fig4}
        \end{figure}

Now, considering all the parameter values specified in table (\ref{t1}), with the exception of $\mu=3.3125$, and by varying the prey refuge parameter $n$, it becomes apparent that the stability of the coexisting fixed point experiences a change near the value $n=0.0555353=n^{ns}$. More precisely, stable periodic cycles arise as a result of the manifestation of a Neimark-Sacker bifurcation at the critical value $n=n^{ns}$. Numerical verification of this claim can be carried out by applying the theorem (\ref{t5}). After some calculations, we get $p_1=-0.969343$, $p_2=-0.985856$, and $p_3=0.983486$. Subsequently, the following calculations are obtained: $1-p_2+ p_3(p_1-p_3) = 0$, $1+p_2-p_3(p_1+p_3)=0.63923 > 0$, $1+p_1+p_2+p_3=0.0403936 > 0$, $1-p_1+p_2-p_3=0.696705 > 0$,  $\frac{d}{dn}(1-p_2+ p_3(p_1-p_3))_{n=n^{ns}}=0.114834 \neq 0$, and by utilising the equation $cos(\frac{2\pi}{j})=0.988932$, it is found that $j=\pm 42.1926$, consequently, it can be concluded that the non-resonance criterion is also met. Therefore, based on theorem (\ref{t5}), it can be concluded that all the necessary conditions for the occurrence of a Neimark-Sacker bifurcation have been satisfied. In addition, the figure (\ref{fig5}) confirms the same.

The parameter related to the prey refuge $n$ plays a crucial role in maintaining the coexistence of all species within the system being examined. Considering the specified parameter values $h=0.1$, $r=0.75$, $\alpha_1=10$, $\alpha_2=10$, $\mu=2.375$, $\alpha_3=0.03$, $\beta=5.3$, $\gamma=9.5$, $k=1$, and $m=3.5$ and manipulating the value of $n$, it is found that when $n=0.384975$, the system (\ref{eq4}) exhibits coexistence of all species. The values of $p_1$, $p_2$, and $p_3$ are computed as $p_1=-1.15274$, $p_2=-0.453689$ and $p_3=0.622965$. Consequently, $1+p_2-|p_1+p_3|=0.0165315>0$, $1-p_3^2-p_2+p_1 p_3=0.347485>0$, and $ 3-p_2-|p_1-3p_3|=0.43205>0$ supports the coexistence of all species. However, when the value of the prey refuge parameter $n$ increases and after a certain value $n$, the axial equilibrium point exhibits stability. At $n=0.884975$, all the eigenvalues of the Jacobian matrix $J_{e2}$ are 0.989261, -0.718981, and 0.885981, i.e., the modulus of all the eigenvalues is less than one, confirming the axial equilibrium's stability.

\begin{figure}[H]
     \centering
     \begin{subfigure}[F]{0.4\textwidth}
         \centering
         \includegraphics[width=1.2\textwidth]{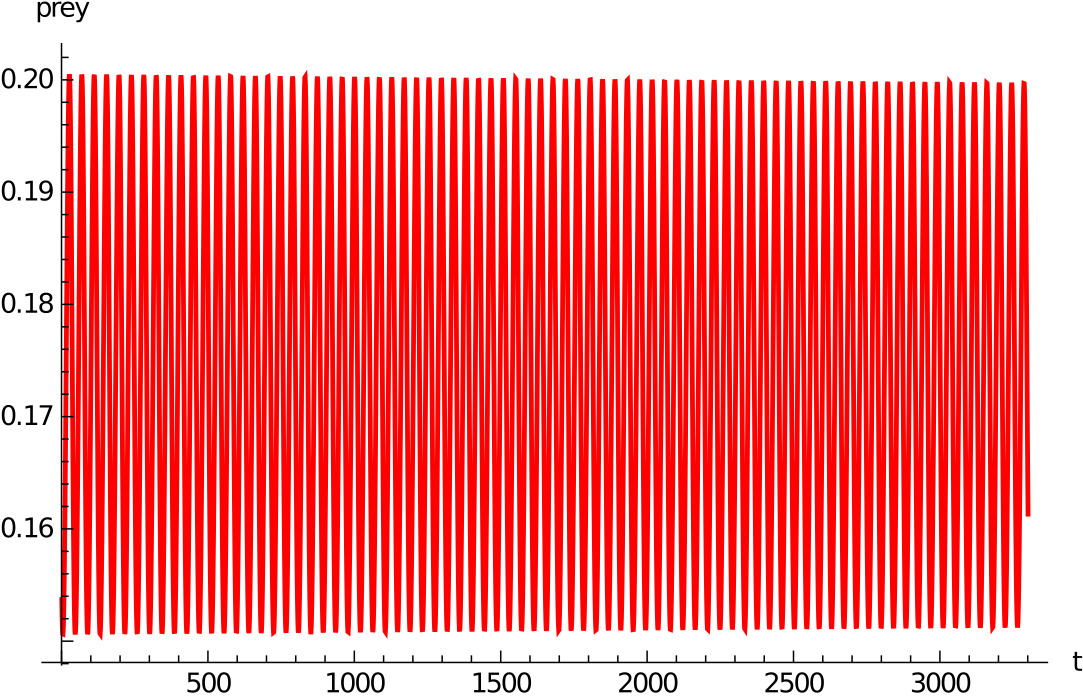}
         \caption{\emph{Time series of prey }}
         \label{fig:y equals x}
     \end{subfigure}
     \hfill
     \begin{subfigure}[F]{0.5\textwidth}
         \centering
         \includegraphics[width=\textwidth]{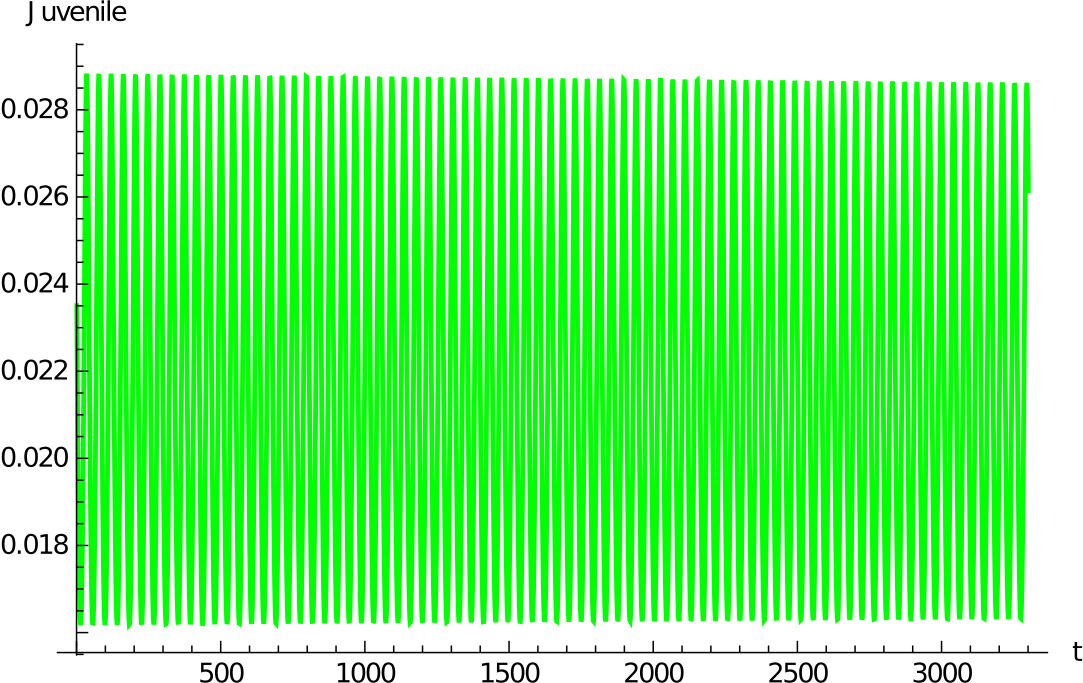}
         \caption{\emph{Time series of juvenile predator}}
     \end{subfigure}

     \begin{subfigure}[F]{0.45\textwidth}
         \centering
         \includegraphics[width=\textwidth]{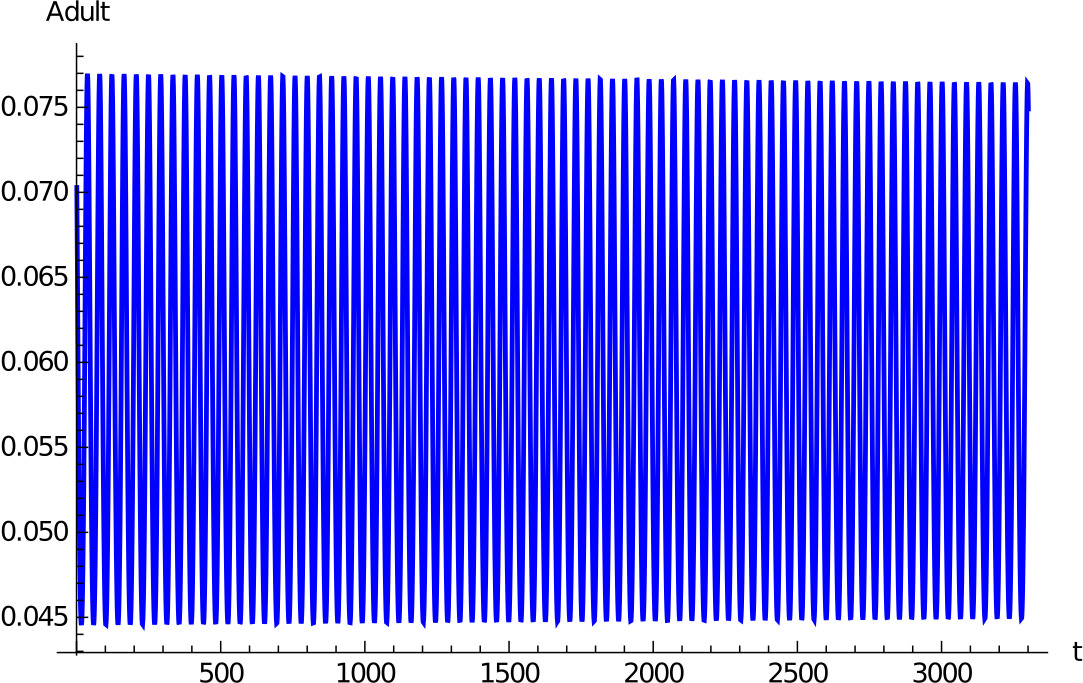}
         \caption{\emph{Time series of adult predator}}
     \end{subfigure} \hfill
      \begin{subfigure}[F]{0.45\textwidth}
         \centering
         \includegraphics[width=\textwidth]{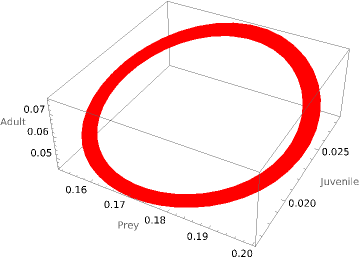}
         \caption{\emph{Phase portrait}}
     \end{subfigure}
        \caption{\emph{The occurrence of Neimark-Sacker bifurcation is shown when the prey refuge parameter $n$ is altered. The figure is generated utilising the parameter values $n=0.055535$ and $mu=3.3125$, while the remaining parameters are maintained in accordance with the values indicated in table (\ref{t1}).}}

        \label{fig5}
        \end{figure}

\section{Conclusion} The primary objective of our study is to investigate a predator-prey model incorporating stage structure with juvenile hunting, paying special attention to the adverse effects of prey refuge behaviour on their own population dynamics, utilising a discrete-time mathematical model that accurately represents this specific scenario, taking into account nonoverlapping generation for the species under consideration. The model under investigation in this paper is a discrete counterpart of the continuous model proposed by Kaushik et al. \cite{motive1} with certain modifications made to the original model. In this paper, the existence conditions of all ecologically relevant fixed points are found and  a stability analysis of these fixed points are done. Under certain parametric conditions, it is observed that both the axial fixed point and the coexisting fixed point exhibit stability. A comprehensive bifurcation analysis is performed. It is found that the population dynamics in this model are significantly influenced by the intrinsic growth rate of the prey ($r$). Different bifurcations of codimension 1, such as the Neimark-Sacker bifurcation and the period-doubling bifurcation, can be observed when the growth rate of the prey-related parameter $r$ is varied.  The figures (\ref{fig3}) and (\ref{fig4}) illustrate the manifestation of the Neimark-Sacker bifurcation and the period-doubling bifurcation, respectively, in relation to the parameter $r$. Furthermore, the parameter $n$, which is related to the prey refuge, plays a crucial role in sustaining population stability within the system under study.  It is observed that varying the value of the parameter $n$ can result in the coexistence of all species or the extermination of predator species. The occurrence of a Neimark-Sacker bifurcation is also observed, leading to the destabilisation of the system when the value of the prey refuge parameter varies, as portrayed in the figure (\ref{fig5}). These illustrate the significance of prey refuge in the system under consideration. Various numerical simulations and graphical representations are presented in this paper to demonstrate the intricate dynamics of the system model.

\end{document}